\documentclass{article}%
\usepackage{amssymb}
\usepackage{amsmath}
\usepackage{amsfonts}
\usepackage{graphicx}%
\providecommand{\U}[1]{\protect\rule{.1in}{.1in}}
\newtheorem{theorem}{Theorem}[section]

\newtheorem{conjecture}[theorem]{Conjecture}
\newtheorem{corollary}[theorem]{Corollary}

\newtheorem{proposition}[theorem]{Proposition}
\newtheorem{remark}[theorem]{Remark}

\newenvironment{proof}[1][Proof]{\noindent\textbf{#1.} }{\ \rule{0.5em}{0.5em}}
\begin{document}

\title{On the Structure of the Minimum Critical Independent Set of a Graph}
\author{Vadim E. Levit\\Ariel University Center of Samaria, Israel\\levitv@ariel.ac.il
\and Eugen Mandrescu\\Holon Institute of Technology, Israel\\eugen\_m@hit.ac.il}
\date{}
\maketitle

\begin{abstract}
Let $G=\left(  V,E\right)  $. A set $S\subseteq V$ is \textit{independent} if
no two vertices from $S$ are adjacent, and by $\mathrm{Ind}(G)$ we mean the
set of all independent sets of $G$. The number $d\left(  X\right)  =$
$\left\vert X\right\vert -\left\vert N(X)\right\vert $ is the
\textit{difference} of $X\subseteq V$, and $A\in\mathrm{Ind}(G)$ is
\textit{critical} if
\[
d(A)=\max\{d\left(  I\right)  :I\in\mathrm{Ind}(G)\}\text{ \cite{Zhang}}.
\]

Let us recall the following definitions:
\[
\mathrm{\ker}(G)=\cap\left\{  S:S\text{\textit{ is a critical independent
set}}\right\}  \text{ \cite{Levman2011a},}%
\]%
\[
\mathrm{core}\left(  G\right)  =\cap\left\{  S:S\text{\textit{ is a maximum
independent set}}\right\}  \text{ \cite{LevMan2002a}}.
\]

Recently, it was established that \ $\mathrm{\ker}(G)\subseteq\mathrm{core}%
(G)$ is true for every graph \cite{Levman2011a}, while the corresponding
equality holds for bipartite graphs \cite{Levman2011b}.

In this paper we present various structural properties of $\mathrm{\ker}(G)$.
The main finding claims that
\[
\mathrm{\ker}(G)=\cup\left\{  S_{0}:S_{0}\text{ is an inclusion minimal
independent set with }d\left(  S_{0}\right)  >0\right\}  .
\]
.


\end{abstract}

\section{Introduction}

Throughout this paper $G=(V,E)$ is a simple (i.e., a finite, undirected,
loopless and without multiple edges) graph with vertex set $V=V(G)$ and edge
set $E=E(G)$. If $X\subseteq V$, then $G[X]$ is the subgraph of $G$ spanned by
$X$. By $G-W$ we mean either the subgraph $G[V-W]$, if $W\subseteq V(G)$, or
the partial subgraph $H=(V,E-W)$ of $G$, for $W\subseteq E(G)$. In either
case, we use $G-w$, whenever $W$ $=\{w\}$.

The \textit{neighborhood} of a vertex $v\in V$ is the set $N(v)=\{w:w\in V$
\ \textit{and} $vw\in E\}$, while the \textit{closed neighborhood} of $v\in V$
is $N[v]=N(v)\cup\{v\}$; in order to avoid ambiguity, we use also $N_{G}(v)$
instead of $N(v)$. The \textit{neighborhood} of $A\subseteq V$ is\emph{
}denoted by $N(A)=N_{G}(A)=\{v\in V:N(v)\cap A\neq\emptyset\}$, and
$N[A]=N(A)\cup A$.

A set $S\subseteq V(G)$ is \textit{independent} if no two vertices from $S$
are adjacent, and by $\mathrm{Ind}(G)$ we mean the set of all the independent
sets of $G$.

An independent set of maximum size will be referred to as a \textit{maximum
independent set} of $G$, and the \textit{independence number }of $G$ is
$\alpha(G)=\max\{\left\vert S\right\vert :S\in\mathrm{Ind}(G)\}$. Let
$\Omega(G)$ denote the family of all maximum independent sets, and
$\mathrm{core}(G)=\cap\{S:S\in\Omega(G)\}$ \cite{LevMan2002a}.

A \textit{matching} is a set of non-incident edges of $G$; a matching of
maximum cardinality is a \textit{maximum matching}, and its size is denoted by
$\mu(G)$.

The number $d(X)=\left\vert X\right\vert -\left\vert N(X)\right\vert $,
$X\subseteq V(G)$, is called the \textit{difference} of the set $X$. The
number $d_{c}(G)=\max\{d(X):X\subseteq V\}$ is called the \textit{critical
difference} of $G$, and a set $U\subseteq V(G)$ is \textit{critical} if
$d(U)=d_{c}(G)$ \cite{Zhang}. The number $id_{c}(G)=\max\{d(I):I\in
\mathrm{Ind}(G)\}$ is called the \textit{critical independence difference} of
$G$. If $A\subseteq V(G)$ is independent and $d(A)=id_{c}(G)$, then $A$ is
called \textit{critical independent }\cite{Zhang}. Clearly, $d_{c}(G)\geq
id_{c}(G)$ is true for every graph $G$.

\begin{theorem}
\label{Theorem3}\cite{Zhang} The equality $d_{c}(G)$ $=id_{c}(G)$ holds for
every graph $G$.
\end{theorem}

For a graph $G$, let denote $\mathrm{\ker}(G)=\cap\left\{  S:S\text{
\textit{is a critical independent set}}\right\}  $. It is known that
\ $\mathrm{\ker}(G)\subseteq\mathrm{core}(G)$ is true for every graph
\cite{Levman2011a}, while the equality holds for bipartite graphs
\cite{Levman2011b}.

For instance, the graph $G$ from Figure \ref{fig51} has $X=\left\{
v_{1},v_{2},v_{3},v_{4}\right\}  $ as a critical set, since $N(X)=\{v_{3}%
,v_{4},v_{5}\}$ and $d(X)=1=d_{c}(G)$, while $I=\{v_{1},v_{2},v_{3}%
,v_{6},v_{7}\}$ is a critical independent set, because $d(I)=1=id_{c}(G)$;
other critical sets are $\{v_{1},v_{2}\}$, $\{v_{1},v_{2},v_{3}\}$,
$\{v_{1},v_{2},v_{3},v_{4},v_{6},v_{7}\}$. In addition, $\mathrm{\ker
}(G)=\{v_{1},v_{2}\}$, and $\mathrm{core}(G)$ is a critical set.
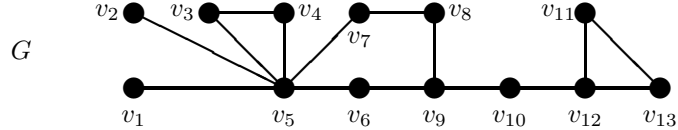
\begin{figure}[h]
\setlength{\unitlength}{1cm}\begin{picture}(5,1.9)\thicklines
\multiput(6,0.5)(1,0){6}{\circle*{0.29}}
\multiput(5,1.5)(1,0){4}{\circle*{0.29}}
\multiput(4,0.5)(0,1){2}{\circle*{0.29}}
\put(10,1.5){\circle*{0.29}}
\put(4,0.5){\line(1,0){7}}
\put(4,1.5){\line(2,-1){2}}
\put(5,1.5){\line(1,-1){1}}
\put(5,1.5){\line(1,0){1}}
\put(6,0.5){\line(0,1){1}}
\put(6,0.5){\line(1,1){1}}
\put(7,1.5){\line(1,0){1}}
\put(8,0.5){\line(0,1){1}}
\put(10,0.5){\line(0,1){1}}
\put(10,1.5){\line(1,-1){1}}
\put(4,0.1){\makebox(0,0){$v_{1}$}}
\put(3.65,1.5){\makebox(0,0){$v_{2}$}}
\put(4.65,1.5){\makebox(0,0){$v_{3}$}}
\put(6.35,1.5){\makebox(0,0){$v_{4}$}}
\put(6,0.1){\makebox(0,0){$v_{5}$}}
\put(7,0.1){\makebox(0,0){$v_{6}$}}
\put(7,1.15){\makebox(0,0){$v_{7}$}}
\put(8,0.1){\makebox(0,0){$v_{9}$}}
\put(8.35,1.5){\makebox(0,0){$v_{8}$}}
\put(9.65,1.5){\makebox(0,0){$v_{11}$}}
\put(9,0.1){\makebox(0,0){$v_{10}$}}
\put(10,0.1){\makebox(0,0){$v_{12}$}}
\put(11,0.1){\makebox(0,0){$v_{13}$}}
\put(2.5,1){\makebox(0,0){$G$}}
\end{picture}\caption{\textrm{core}$(G)=\{v_{1},v_{2},v_{6},v_{10}\}$.}%
\label{fig51}%
\end{figure}

It is easy to see that all pendant vertices are included in every maximum
critical independent set. It is known that the problem of finding a critical
independent set is polynomially solvable \cite{Ageev,Zhang}.

\begin{theorem}
\label{th7}For a graph $G=(V,E)$, the following assertions are true:

\emph{(i)} \cite{Levman2011a} the function $d$ is supermodular, i.e.,
\[
d(A\cup B)+d(A\cap B)\geq d(A)+d(B)\text{ for every }A,B\subseteq V;
\]

\emph{(ii)} \cite{Levman2011a} $G$ has a unique minimal critical independent
set, namely, $\mathrm{\ker}(G)$.

\emph{(iii)} \cite{Larson} there is a matching from $N(S)$ into $S$, for every
critical independent set $S$.
\end{theorem}

In this paper we characterize $\mathrm{\ker}(G)$. In addition, a number of
properties of $\mathrm{\ker}(G)$\ are presented as well.

\section{Results}

Deleting a vertex from a graph may decrease, leave unchanged or increase its
critical difference. For instance, $d_{c}\left(  G-v_{1}\right)  =d_{c}\left(
G\right)  -1$, $d_{c}\left(  G-v_{13}\right)  =d_{c}\left(  G\right)  $, while
$d_{c}\left(  G-v_{3}\right)  =d_{c}\left(  G\right)  +1$, where $G$ is
depicted in Figure \ref{fig51}.

\begin{proposition}
Let $G=(V,E)$ and $v\in V$. Then the following assertions hold:

\emph{(i) }$d_{c}\left(  G-v\right)  =d_{c}\left(  G\right)  -1$ if and only
if $v\in\mathrm{\ker}(G)$;

\emph{(ii)} if $v\in\mathrm{\ker}(G)$, then $\mathrm{\ker}(G-v)\subseteq
\mathrm{\ker}(G)-\left\{  v\right\}  $.
\end{proposition}

\begin{proof}
\emph{(i) }Let $v\in V$ and $H=G-v$.

If $v\notin\mathrm{\ker}(G)$, then $\mathrm{\ker}(G)\subseteq V\left(
G\right)  -\left\{  v\right\}  $. Hence
\[
d_{c}\left(  G-v\right)  \geq\left\vert \mathrm{\ker}(G)\right\vert
-\left\vert N_{H}\left(  \mathrm{\ker}(G)\right)  \right\vert \geq\left\vert
\mathrm{\ker}(G)\right\vert -\left\vert N_{G}\left(  \mathrm{\ker}(G)\right)
\right\vert =d_{c}\left(  G\right)  .
\]
Consequently, we infer that $d_{c}\left(  G-v\right)  <d_{c}\left(  G\right)
$ implies $v\in\mathrm{\ker}(G)$.

Conversely, assume that $v\in\mathrm{\ker}(G)$. Each $u\in N\left(  v\right)
$ satisfies $\left\vert N\left(  u\right)  \cap\mathrm{\ker}(G)\right\vert
\geq2$, because otherwise, $d\left(  \mathrm{\ker}(G)-\{v\}\right)  =d\left(
\mathrm{\ker}(G)\right)  $ and this contradicts the minimality of
$\mathrm{\ker}(G)$. Therefore, $N\left(  \mathrm{\ker}(G)-\{v\}\right)
=N\left(  \mathrm{\ker}(G)\right)  $ and hence
\begin{gather*}
d\left(  \mathrm{\ker}(G)-\{v\}\right)  =\left\vert \mathrm{\ker
}(G)-\{v\}\right\vert -\left\vert N\left(  \mathrm{\ker}(G)-\{v\}\right)
\right\vert =\\
=\left\vert \mathrm{\ker}(G)\right\vert -1-\left\vert N\left(  \mathrm{\ker
}(G)\right)  \right\vert =d_{c}\left(  G\right)  -1.
\end{gather*}
If there is some independent set $A$ in $G-v$, such that $d(A)=d_{c}\left(
G\right)  $, then $A$ is critical in $G$ and, hence we get the following
contradiction: $v\in\mathrm{\ker}(G)\subseteq A\subseteq V-\left\{  v\right\}
$. Therefore, $\mathrm{\ker}(G)-\{v\}$ is a critical independent set of $G-v$
and
\[
d_{c}\left(  G-v\right)  =d\left(  \mathrm{\ker}(G)-\{v\}\right)
=d_{c}\left(  G\right)  -1.
\]

\emph{(ii)} Assume that $\mathrm{\ker}(G-v)\neq\emptyset$. In part \emph{(i)},
we saw that $\mathrm{\ker}(G)-\{v\}$ is a critical independent set of $G-v$.
Hence, we get that $\mathrm{\ker}(G-v)\subseteq\mathrm{\ker}(G)-\left\{
v\right\}  $.
\end{proof}

\begin{remark}
Actually, $\mathrm{\ker}(G-v)$ may be different from $\mathrm{\ker
}(G)-\left\{  v\right\}  $; for instance, if $K_{3,2}=(A,B,E)$, $\left\vert
A\right\vert =3$, then $\mathrm{\ker}(K_{3,2})=A$ and $\mathrm{\ker}%
(K_{3,2}-v)=\emptyset\neq\mathrm{\ker}(K_{3,2})-\left\{  v\right\}  $, for
every $v\in A$. It is also possible $\mathrm{\ker}(G)-\left\{  v\right\}
=\emptyset$, while $\mathrm{\ker}(G-v)\neq\emptyset$; e.g., $G=C_{4}$.
\end{remark}

By Theorem \ref{th7}\emph{(iii)}, there is a matching from $N\left(  S\right)
$ into $S=\left\{  v_{1},v_{2},v_{3}\right\}  $, for instance, $M=\left\{
v_{2}v_{5},v_{3}v_{4}\right\}  $, since $S$ is critical independent for the
graph $G$ from Figure \ref{fig51}. On the other hand, there is no matching
from $N\left(  S\right)  $ into $S-v_{3}$. The case of the critical
independence set $\mathrm{\ker}(G)$ is more specific.

\begin{theorem}
\label{th9}Let $A$ be a critical independent set in a graph $G$. Then the
following statements are equivalent:

\emph{(i) }$A=\mathrm{\ker}(G)$;

\emph{(ii)} there is no set $B\subseteq$ $N\left(  A\right)  ,B\neq\emptyset$
such that $\left\vert N\left(  B\right)  \cap A\right\vert =\left\vert
B\right\vert $;

\emph{(iii) }for each $v\in A$ there exists a matching from $N\left(
A\right)  $ into $A-v$.
\end{theorem}

\begin{proof}
\emph{(i) }$\Longrightarrow$ \emph{(ii) }By Theorem \ref{th7}\emph{(iii)},
there is a matching, say $M$, from $N\left(  \mathrm{\ker}(G)\right)  $ into
$\mathrm{\ker}(G)$. Suppose, to the contrary, that there is some non-empty
set$\ B\subseteq$ $N\left(  \mathrm{\ker}(G)\right)  $ such that
\[
\left\vert M\left(  B\right)  \right\vert =\left\vert N\left(  B\right)
\cap\mathrm{\ker}(G)\right\vert =\left\vert B\right\vert .
\]
It contradicts the fact that, by Theorem \ref{th7}\emph{(ii)}, $\mathrm{\ker
}(G)$ is a minimal critical independent set, because
\[
d\left(  \mathrm{\ker}(G)-N\left(  B\right)  \right)  =d\left(  \mathrm{\ker
}(G)\right)  \text{, while }\mathrm{\ker}(G)-N\left(  B\right)  \subsetneqq
\mathrm{\ker}(G)\text{.}%
\]

\emph{(ii) }$\Longrightarrow$ \emph{(i) }Suppose $A-\mathrm{\ker}%
(G)\neq\emptyset$. By Theorem \ref{th7}\emph{(iii)}, there is a matching, say
$M$, from $N\left(  A\right)  $ into $A$. Since there are no edges connecting
vertices belonging to $\mathrm{\ker}(G)$ with vertices from $N\left(
A\right)  -N\left(  \mathrm{\ker}(G)\right)  $, we obtain that $M\left(
N\left(  A\right)  -N\left(  \mathrm{\ker}(G)\right)  \right)  \subseteq
A-\mathrm{\ker}(G)$. Moreover, we have that $\left\vert N\left(  A\right)
-N\left(  \mathrm{\ker}(G)\right)  \right\vert =\left\vert A-\mathrm{\ker
}(G)\right\vert $, otherwise
\begin{align*}
\left\vert A\right\vert -\left\vert N\left(  A\right)  \right\vert  &
=\left(  \left\vert \mathrm{\ker}(G)\right\vert -\left\vert N\left(
\mathrm{\ker}(G)\right)  \right\vert \right)  +\left(  \left\vert
A-\mathrm{\ker}(G)\right\vert -\left\vert N\left(  A\right)  -N\left(
\mathrm{\ker}(G)\right)  \right\vert \right)  >\\
&  >\left(  \left\vert \mathrm{\ker}(G)\right\vert -\left\vert N\left(
\mathrm{\ker}(G)\right)  \right\vert \right)  =d_{c}\left(  G\right)  .
\end{align*}

It means that the set $N\left(  A\right)  -N\left(  \mathrm{\ker}(G)\right)  $
contradicts the hypothesis of \emph{(ii)}, because
\[
\left\vert N\left(  A\right)  -N\left(  \mathrm{\ker}(G)\right)  \right\vert
=\left\vert A-\mathrm{\ker}(G)\right\vert =\left\vert N\left(  N\left(
A\right)  -N\left(  \mathrm{\ker}(G)\right)  \right)  \cap A\right\vert
\text{.}%
\]
Consequently, the assertion is true.

\emph{(ii) }$\Longrightarrow$ \emph{(iii)} By Theorem \ref{th7}\emph{(iii)},
there is a matching, say $M$, from $N\left(  A\right)  $ into $A$. Suppose, to
the contrary, that there is no matching from $N\left(  A\right)  $ into $A-v$.
Hence, by Hall's Theorem, it implies the existence of a set $B\subseteq$
$N\left(  A\right)  $ such that $\left\vert N\left(  B\right)  \cap
A\right\vert =\left\vert B\right\vert $, which contradicts the hypothesis of
\emph{(ii)}.

\emph{(iii) }$\Longrightarrow$ \emph{(ii)} Assume, to the contrary, that there
is a non-empty subset $B$ of $N\left(  A\right)  $ such that $\left\vert
N\left(  B\right)  \cap A\right\vert =\left\vert B\right\vert $. Let $v$ $\in
N\left(  B\right)  \cap A$. Hence, we obtain that
\[
\left\vert N\left(  B\right)  \cap A-v\right\vert <\left\vert B\right\vert .
\]
Then, by Hall's Theorem, it is impossible to find a matching from $N\left(
A\right)  $ into $A-v$, in contradiction with the hypothesis of \emph{(iii)}.
\end{proof}

Since $\mathrm{\ker}(G)$ is a critical set, Theorem \ref{th7}\emph{(iii)}
assures that there is a matching from $N\left(  \mathrm{\ker}(G)\right)  $
into $\mathrm{\ker}(G)$. The following result shows that there are at least
two such matchings.

\begin{corollary}
For a graph $G$ the following are true:

\emph{(i)} every edge $e\in\left(  \mathrm{\ker}(G),N\left(  \mathrm{\ker
}(G)\right)  \right)  $ belongs to a matching from $N\left(  \mathrm{\ker
}(G)\right)  $ into $\mathrm{\ker}(G)$;

\emph{(ii)} every edge $e\in\left(  \mathrm{\ker}(G),N\left(  \mathrm{\ker
}(G)\right)  \right)  $ is not included in one matching from $N\left(
\mathrm{\ker}(G)\right)  $ into $\mathrm{\ker}(G)$ at least.
\end{corollary}

\begin{proof}
Let $e=xy\in\left(  \mathrm{\ker}(G),N\left(  \mathrm{\ker}(G)\right)
\right)  $, such that $x\in\mathrm{\ker}(G)$. By Theorem \ref{th9}\emph{(iii)}
there is a matching $M$ from $N\left(  \mathrm{\ker}(G)\right)  $ into
$\mathrm{\ker}(G)-x$, that matches $y$ with some $z\in\mathrm{\ker}(G)-x$.
Clearly, $M$ is a matching from $N\left(  \mathrm{\ker}(G)\right)  $ into
$\mathrm{\ker}(G)$ that does not contain the edge $e=xy$, while $\left(
M-\left\{  yz\right\}  \right)  \cup\left\{  xy\right\}  $ is a matching from
$N\left(  \mathrm{\ker}(G)\right)  $ into $\mathrm{\ker}(G)$, which includes
the edge $e=xy$.
\end{proof}

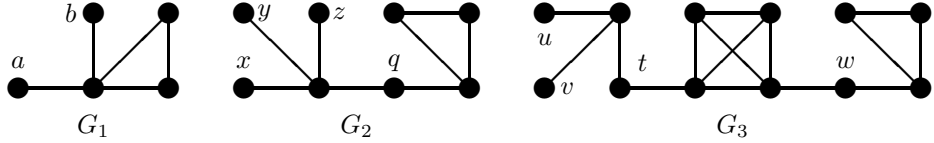
\begin{figure}[h]
\setlength{\unitlength}{1.0cm} \begin{picture}(5,1.9)\thicklines
\multiput(1,0.5)(1,0){3}{\circle*{0.29}}
\multiput(2,1.5)(1,0){2}{\circle*{0.29}}
\put(1,0.5){\line(1,0){2}}
\put(2,0.5){\line(0,1){1}}
\put(2,0.5){\line(1,1){1}}
\put(3,0.5){\line(0,1){1}}
\put(1,0.84){\makebox(0,0){$a$}}
\put(1.7,1.5){\makebox(0,0){$b$}}
\put(2,0){\makebox(0,0){$G_{1}$}}
\multiput(4,0.5)(1,0){4}{\circle*{0.29}}
\put(4,0.5){\line(1,0){3}}
\multiput(4,1.5)(1,0){4}{\circle*{0.29}}
\put(4,1.5){\line(1,-1){1}}
\put(5,0.5){\line(0,1){1}}
\put(7,0.5){\line(0,1){1}}
\put(7,0.5){\line(-1,1){1}}
\put(6,1.5){\line(1,0){1}}
\put(6,0.84){\makebox(0,0){$q$}}
\put(4,0.84){\makebox(0,0){$x$}}
\put(4.27,1.5){\makebox(0,0){$y$}}
\put(5.27,1.5){\makebox(0,0){$z$}}
\put(5.5,0){\makebox(0,0){$G_{2}$}}
\multiput(8,0.5)(1,0){6}{\circle*{0.29}}
\multiput(8,1.5)(1,0){6}{\circle*{0.29}}
\put(8,0.5){\line(1,1){1}}
\put(8,1.5){\line(1,0){1}}
\put(9,0.5){\line(0,1){1}}
\put(9,0.5){\line(1,0){4}}
\put(10,0.5){\line(0,1){1}}
\put(10,0.5){\line(1,1){1}}
\put(10,1.5){\line(1,-1){1}}
\put(10,1.5){\line(1,0){1}}
\put(11,0.5){\line(0,1){1}}
\put(12,1.5){\line(1,-1){1}}
\put(12,1.5){\line(1,0){1}}
\put(13,0.5){\line(0,1){1}}
\put(8.3,0.5){\makebox(0,0){$v$}}
\put(8,1.15){\makebox(0,0){$u$}}
\put(9.3,0.84){\makebox(0,0){$t$}}
\put(12,0.84){\makebox(0,0){$w$}}
\put(10.5,0){\makebox(0,0){$G_{3}$}}
\end{picture}\caption{$\mathrm{core}(G_{1})=\left\{  a,b\right\}  $,
$\mathrm{core}(G_{2})=\left\{  q,x,y,z\right\}  $, $\mathrm{core}%
(G_{3})=\left\{  t,u,v,w\right\}  $.}%
\label{fig333}%
\end{figure}

Let us notice that the graphs $G_{1}$, $G_{2}$ from Figure \ref{fig333} have:
$\mathrm{\ker}(G_{1})=\mathrm{core}(G_{1})$, $\mathrm{\ker}(G_{2})=\left\{
x,y,z\right\}  \subset\mathrm{core}(G_{2})$, and both $\mathrm{core}(G_{1})$
and $\mathrm{core}(G_{2})$\ are critical sets of maximum size. The graph
$G_{3}$ from Figure \ref{fig333} has $\mathrm{\ker}(G_{3})=\{u,v\}$, the set
$\{t,u,v\}$\ as a critical independent set of maximum size, while
$\mathrm{core}(G_{3})=\left\{  t,u,v,w\right\}  $\ is not a critical set. If
$S_{\min}$ denotes an inclusion minimal independent set with $d\left(
S_{\min}\right)  >0$, one can see that: $S_{\min}=\mathrm{\ker}(G_{1})$ for
$G_{1}$, while the graph $G_{2}$ in the same figure has $S_{\min}\in\left\{
\{x,y\},\{x,z\},\{y,z\}\right\}  $ and $\mathrm{\ker}(G_{2})=\{x,y\}\cup
\{x,z\}\cup\{y,z\}$.

In \cite{Levman2011a}\ we have shown that $\mathrm{\ker}(G)$ is equal to the
intersection of all critical, independent or not, sets of $G$.

\begin{theorem}
\label{th1}For every graph $G$
\[
\mathrm{\ker}(G)=\cup\left\{  S_{0}:S_{0}\text{ is an inclusion minimal
independent set with }d\left(  S_{0}\right)  >0\right\}  .
\]

\end{theorem}

\begin{proof}
Let $A$ be a critical set and $S_{0}$ be an inclusion minimal independent set
such that $d\left(  S_{0}\right)  >0$. Then, Theorem \ref{th7}\emph{(i)}
implies
\[
d(A\cup S_{0})+d(A\cap S_{0})\geq d(A)+d(S_{0})>d(A)=d_{c}\left(  G\right)  .
\]
Since $S_{0}$ is an inclusion minimal independent set such that $d(S_{0})>0$,
we obtain that if $A\cap S_{0}\neq S_{0}$, then $d(A\cap S_{0})\leq0$. Hence
\[
d(A)=d_{c}\left(  G\right)  \geq d(A\cup S_{0})\geq d(A)+d(S_{0})>d(A),
\]
which is impossible. Therefore, $S_{0}\subseteq A$ for every critical set $A$.
Consequently,%
\[
S_{0}\subseteq\cap\left\{  B:B\text{ \textit{is a critical set of} }G\right\}
=\mathrm{\ker}(G).
\]
Thus we obtain%
\[
\cup\left\{  S_{0}:S_{0}\text{ \textit{is an inclusion minimal independent set
such that} }d\left(  S_{0}\right)  >0\right\}  \subseteq\mathrm{\ker}(G).
\]

Conversely, it is enough to show that every vertex from $\mathrm{\ker}(G)$
belongs to some inclusion minimal independent set with positive difference.
Let $v\in\mathrm{\ker}(G)$. According to Theorem \ref{th9}\emph{(iii)} there
exists a matching, say $M$, from $N\left(  \mathrm{\ker}(G)\right)  $ into
$\mathrm{\ker}(G)-v$.

Let us build the following sequence of sets
\[
\left\{  v\right\}  \subseteq M\left(  N\left(  v\right)  \right)
\subseteq...\subseteq\left[  MN\right]  ^{k}\left(  v\right)  \subseteq...,
\]
where $MN$ is a superposition of two mappings $N:2^{V}\longrightarrow2^{V}$
($N\left(  A\right)  $ is the neighborhood of $A$) and $M:2^{N\left(
\mathrm{\ker}(G)\right)  }\longrightarrow2^{\mathrm{\ker}(G)}$ ($M\left(
A\right)  $ is set of the vertices matched by $M$ with vertices belonging to
$A$).

Since the set $\mathrm{\ker}(G)$ is finite, there is an index $j$ such that
$\left[  MN\right]  ^{j}\left(  v\right)  =\left[  MN\right]  ^{j+1}\left(
v\right)  $. Hence $\left\vert N\left(  \left[  MN\right]  ^{j}\left(
v\right)  \right)  \right\vert =\left\vert \left[  MN\right]  ^{j}\left(
v\right)  \right\vert -1$. In other words, we found an independent set,
namely, $\left[  MN\right]  ^{j}\left(  v\right)  $ such that $v\in\left[
MN\right]  ^{j}\left(  v\right)  $ and $d\left(  \left[  MN\right]
^{j}\left(  v\right)  \right)  =1$. Therefore, there must exist an inclusion
minimal independent set $X$ such that $v\in X$ and $d\left(  X\right)  =1$.
\end{proof}

\begin{remark}
In a graph $G$, the union of all minimum cardinality independent sets $S$ with
$d\left(  S\right)  >0$ may be a proper subset of $\mathrm{\ker}\left(
G\right)  $; e.g., the graph $G$ in Figure \ref{fig177}, that has $\left\{
x,y\right\}  \subset\mathrm{\ker}\left(  G\right)  =\left\{
x,y,u,v,w\right\}  $.
\end{remark}

\begin{figure}[h]
\setlength{\unitlength}{1cm}\begin{picture}(5,1.3)\thicklines
\multiput(4,0)(1,0){3}{\circle*{0.29}}
\multiput(3,1)(1,0){5}{\circle*{0.29}}
\put(4,0){\line(1,0){2}}
\put(4,0){\line(0,1){1}}
\put(3,1){\line(1,-1){1}}
\put(5,0){\line(0,1){1}}
\put(5,0){\line(2,1){2}}
\put(6,0){\line(0,1){1}}
\put(6,0){\line(1,1){1}}
\put(2.7,1){\makebox(0,0){$x$}}
\put(3.7,1){\makebox(0,0){$y$}}
\put(4.7,1){\makebox(0,0){$u$}}
\put(5.7,1){\makebox(0,0){$v$}}
\put(7.3,1){\makebox(0,0){$w$}}
\put(2,0.5){\makebox(0,0){$G$}}
\multiput(9,0)(1,0){3}{\circle*{0.29}}
\multiput(9,1)(2,0){2}{\circle*{0.29}}
\put(9,0){\line(1,0){2}}
\put(9,1){\line(1,-1){1}}
\put(10,0){\line(1,1){1}}
\put(8.65,0){\makebox(0,0){$v_{1}$}}
\put(8.65,1){\makebox(0,0){$v_{2}$}}
\put(11.35,1){\makebox(0,0){$v_{3}$}}
\put(11.35,0){\makebox(0,0){$v_{4}$}}
\put(8,0.5){\makebox(0,0){$H$}}
\end{picture}\caption{Both $S_{1}=\{x,y\}$ and $S_{2}=\{u,v,w\}$ are inclusion
minimal independent sets satisfying $d\left(  S\right)  >0$.}%
\label{fig177}%
\end{figure}
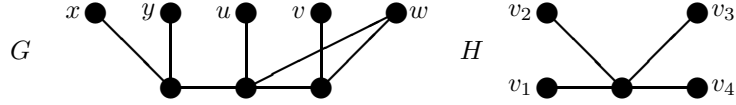

\begin{proposition}
$\min\left\{  \left\vert S_{0}\right\vert :d\left(  S_{0}\right)  >0,S_{0}%
\in\mathrm{Ind}(G)\right\}  \leq\left\vert \mathrm{\ker}\left(  G\right)
\right\vert -d_{c}\left(  G\right)  +1$.
\end{proposition}

\begin{proof}
Since $\mathrm{\ker}(G)$ is a critical independent set, Theorem \ref{th7}%
\emph{(iii)} implies that there is a matching, say $M$, from $N\left(
\mathrm{\ker}(G)\right)  $ into $\mathrm{\ker}(G)$. Let $X=M\left(  N\left(
\mathrm{\ker}(G)\right)  \right)  $. Then $d\left(  X\right)  =0$. For every
$v\in\mathrm{\ker}(G)-X$ we have
\[
N\left(  \mathrm{\ker}(G)\right)  \subseteq N\left(  X\right)  \subseteq
N\left(  X\cup\left\{  v\right\}  \right)  \subseteq N\left(  \mathrm{\ker
}(G)\right)  .
\]
Hence we get $\left\vert X\cup\left\{  v\right\}  \right\vert -\left\vert
N\left(  X\cup\left\{  v\right\}  \right)  \right\vert =1$, while $\left\vert
X\cup\left\{  v\right\}  \right\vert =\left\vert \mathrm{\ker}\left(
G\right)  \right\vert -d_{c}\left(  G\right)  +1$.
\end{proof}

\begin{remark}
All the inclusion minimal independent sets $S$, with $d\left(  S\right)  >0$,
of the graph $H$ from Figure \ref{fig177} are of the same size. However, there
are inclusion minimal independent sets $S$ with $d\left(  S\right)  >0$, of
different cardinalities; e.g., the graph $G$ from Figure \ref{fig177}.
\end{remark}

\begin{proposition}
\label{prop3}If $S_{0}$ is an inclusion minimal independent set with $d\left(
S_{0}\right)  >0$, then $d\left(  S_{0}\right)  =1$.
\end{proposition}

\begin{proof}
For each $v\in S_{0}$, it follows that $N\left(  S_{0}-v\right)  =N\left(
S_{0}\right)  $, otherwise,%
\begin{gather*}
d\left(  S_{0}-v\right)  =\left\vert S_{0}-v\right\vert -\left\vert N\left(
S_{0}-v\right)  \right\vert =\\
=\left\vert S_{0}\right\vert -1-\left\vert N\left(  S_{0}-v\right)
\right\vert \geq\left\vert S_{0}\right\vert -\left\vert N\left(  S_{0}\right)
\right\vert >0\text{,}%
\end{gather*}
i.e., $S_{0}$ is not an inclusion minimal independent set with positive difference.

Since $S_{0}$ is an inclusion minimal independent set with positive
difference, we know that $d\left(  S_{0}-v\right)  \leq0$. On the other hand,
it follows from the equality $N\left(  S_{0}-v\right)  =N\left(  S_{0}\right)
$ that%
\[
d\left(  S_{0}-v\right)  =\left\vert S_{0}-v\right\vert -\left\vert N\left(
S_{0}-v\right)  \right\vert =\left\vert S_{0}\right\vert -1-\left\vert
N\left(  S_{0}\right)  \right\vert =d\left(  S_{0}\right)  -1\leq0\text{.}%
\]

Consequently, $0<\left\vert S_{0}\right\vert -\left\vert N\left(
S_{0}\right)  \right\vert \leq1$, which means that $\left\vert S_{0}%
\right\vert -\left\vert N\left(  S_{0}\right)  \right\vert =1$.
\end{proof}

\begin{remark}
The converse of Proposition \ref{prop3} is not true. For instance, $S=\left\{
x,y,u\right\}  $ is independent in the graph $G$ from Figure \ref{fig177}\ and
$d\left(  S\right)  =1$, but $S$ is not minimal with this property.
\end{remark}

\begin{proposition}
\label{prop2}If $S_{i},i=1,2,...,k,k\geq1$, are inclusion minimal independent
sets, such that $d\left(  S_{i}\right)  >0,S_{i}\nsubseteq\bigcup
\limits_{j=1,j\neq i}^{k}S_{j}$,$1\leq i\leq k$, then $d\left(  S_{1}\cup
S_{2}\cup...\cup S_{k}\right)  \geq k$.
\end{proposition}

\begin{proof}
For $k=1$ the claim has been treated in Proposition \ref{prop3}, where we have
achieved a stronger result.

We continue by induction on $k$.

Let $k=2$. Since $S_{1}\neq S_{1}\cap S_{2}\subset S_{1}$, it follows that
$d\left(  S_{1}\cap S_{2}\right)  \leq0$. Hence, Theorem \ref{th7}\emph{(i)}
and Proposition \ref{prop3} imply
\[
d\left(  S_{1}\cup S_{2}\right)  \geq d\left(  S_{1}\cup S_{2}\right)
+d\left(  S_{1}\cap S_{2}\right)  \geq d\left(  S_{1}\right)  +d\left(
S_{2}\right)  =2.
\]

Assume that the assertion is true for each $k\geq2$, and let $\left\{
S_{i},1\leq i\leq k+1\right\}  $ be a family of inclusion minimal independent
sets with
\[
d\left(  S_{i}\right)  >0\text{ and }S_{i}\nsubseteq\bigcup\limits_{j=1,j\neq
i}^{k+1}S_{j},1\leq i\leq k+1.
\]

Since $S_{k+1}\neq\left(  S_{1}\cup S_{2}\cup...\cup S_{k}\right)  \cap
S_{k+1}\subset$ $S_{k+1}$, we obtain that
\[
d\left(  \left(  S_{1}\cup S_{2}\cup...\cup S_{k}\right)  \cap S_{k+1}\right)
\leq0.
\]
Further, using the supermodularity of the function $d$ and Proposition
\ref{prop3}, we get%
\begin{gather*}
d\left(  S_{1}\cup S_{2}\cup...\cup S_{k}\cup S_{k+1}\right)  \geq\\
\geq d\left(  S_{1}\cup S_{2}\cup...\cup S_{k}\cup S_{k+1}\right)  +d\left(
\left(  S_{1}\cup S_{2}\cup...\cup S_{k}\right)  \cap S_{k+1}\right)  \geq\\
\geq d\left(  S_{1}\cup S_{2}\cup...\cup S_{k}\right)  +d\left(
S_{k+1}\right)  \geq k+1,
\end{gather*}
as required.
\end{proof}

\begin{remark}
The sets $S_{1}=\left\{  v_{1},v_{2}\right\}  ,S_{2}=\left\{  v_{2}%
,v_{3}\right\}  ,S_{3}=\left\{  v_{3},v_{4}\right\}  $ are inclusion minimal
independent sets of the graph $H$ from Figure \ref{fig177}, such that
\[
d\left(  S_{i}\right)  >0,S_{i}\nsubseteq\bigcup\limits_{j=1,j\neq i}^{3}%
S_{j},i=1,2,3.
\]
Notice that both families $\left\{  S_{1},S_{2}\right\}  $, $\left\{
S_{1},S_{3}\right\}  $ have two elements, and $d\left(  S_{1}\cup
S_{2}\right)  =2$, while $d\left(  S_{1}\cup S_{3}\right)  >2$.
\end{remark}

\section{Conclusions}

In this paper we investigate structural properties of $\mathrm{\ker}(G)$.

Having in view Theorem \ref{th1}, notice that the graph:

\begin{itemize}
\item $G_{1}$ from Figure \ref{fig333} has only one inclusion minimal
independent set $S$ such that $d\left(  S\right)  >0$, and $d_{c}\left(
G_{1}\right)  =1$;

\item $G$ from Figure \ref{fig177} has only two inclusion minimal independent
sets $S$ such that $d\left(  S\right)  >0$, and $d_{c}\left(  G\right)  =2$;

\item $H$ from Figure \ref{fig177} has $6$ inclusion minimal independent sets
$S$ such that $d\left(  S\right)  >0$, and $d_{c}\left(  H\right)  =3$.
\end{itemize}

These remarks motivate the following.

\begin{conjecture}
The number of inclusion minimal independent set $S$ such that $d\left(
S\right)  >0$ is greater or equal to $d_{c}\left(  G\right)  $.
\end{conjecture}

\end{document}